\documentclass[12pt]{article}

\usepackage{graphicx}%
\usepackage{multirow}%
\usepackage{amsmath,amssymb,amsfonts}%
\usepackage{amsthm}% 
\usepackage{bm}%
\usepackage{mathrsfs}%
\usepackage{xcolor}%
\usepackage{textcomp}%
\usepackage{manyfoot}%
\usepackage{booktabs}%
\usepackage{algorithm}%
\usepackage{algorithmicx}%
\usepackage{algpseudocode}%
\usepackage{listings}%

\newtheorem{theorem}{Theorem}
\newtheorem{definition}{Definition}%

\begin{document}
\title{Bregman Consensus}

\author{Andrei N. Soklakov\footnote{
%Department of Computing, Security and Mathematics, Royal Holloway, University of London,
%Egham, Surrey TW20 0EX, UK. Email: 
Andrei.Soklakov@gmail.com}}

\date{}

\maketitle

\begin{center}
\parbox{14cm}{
{\small
Consider a community of agents who are seeking consensus on a set of parameters. The agents agree to use the same Bregman-type divergence to quantify disagreement between their individual estimates of the parameters but have varying confidence in each other's abilities. Each agent is happy to revise their estimate by moving to the weighted barycenter of all individual estimates with higher weights applied to more trusted agents. We show that such revisions naturally lead to an iterative algorithm which converges to a unique consensus estimate of the parameters. Furthermore, since the consensus estimate is itself a barycenter with computable weights, the group emerges as a collective super-agent with a well-formed opinion regarding the ability of each individual agent.

%{\bf Keywords:} Multi-agent consensus, Opinion dynamics, Bregman barycenter, Non-scalar means

%{\bf JEL Classification:} C02, C61, C62, D83, D85

%{\bf MSC Classification:} 62B11, 93D50, 26E60, 90C25, 94A17

}}
\end{center}

\newpage

\section{Introduction}

Information aggregation from multiple sources is an incredibly rich field with many practical applications. The mathematical roots of the problem lie in two distinct traditions. In the statistical literature, DeGroot~\cite{DeGroot_1974} considered a group of agents who repeatedly update their estimates as a weighted arithmetic average of the group's previous estimates to achieve consensus. This linear consensus framework has since become a standard tool in distributed averaging and opinion dynamics. In parallel, the theory of quasi-arithmetic means, developed independently by Kolmogorov~\cite{Kolmogorov_1930}, Nagumo~\cite{Nagumo_1930}, and de Finetti~\cite{deFinetti_1931}, characterized the class of scalar means that can be written as $M_f(x_1,\dots,x_N)=f^{-1}(\sum_i w_i f(x_i))$ for a continuous strictly monotone function $f$. Matkowski~\cite{Matkowski_1999} extended this to iterations of mean-type mappings, proving convergence to invariant scalar means. 

The present paper sits at the intersection of these traditions. It replaces the arithmetic average in DeGroot's iteration with a Bregman barycenter, so that agents with varying trust in one another's abilities converge to a consensus that is itself a weighted Bregman barycenter of the initial estimates. In doing so, it provides a geometric divergence-based nonlinear generalisation of the linear consensus model. At the same time, it provides a partial generalisation of Matkowski's convergence result to the case of non-scalar means~\cite{Nielsen_2023} by providing a concrete analytical example.

We use the rest of this introduction to fix basic notation by recalling the classic notion of Bregman barycenters. Section~\ref{Sec:center-seeking_agents} introduces center-seeking agents and provides the main theorem regarding the collective evolution of such agents towards a unique consensus. In Section~\ref{Sec:GroupAgent} we look at the emergent group behaviour as a whole. In Section~\ref{Sec:ConvergenceSpeed} we estimate the speed of convergence. Finally, we discuss the connection to the theory of generalised means in Section~\ref{Sec:NonScalarMeans}.

Let $\bm{x}$ be a $d$-dimensional vector of parameters and let $\mathbb{X}\subset\mathbb{R}^d$ be an open convex set containing all possible values of $\bm{x}$. The Bregman divergence~\cite{Bregman_1967} between two vectors $\bm{x}$ and $\bm{x}'$ is defined as
\begin{equation}\label{Eq:BregmanDiv}
D_F(\bm{x},\bm{x}')=F(\bm{x})-F(\bm{x}')-\langle \bm{x}-\bm{x}',\nabla F(\bm{x}')\rangle\,,
\end{equation} 
where $F:\mathbb{X}\to\mathbb{R}$ is a strictly convex differentiable function. We furthermore require that the gradient map $\nabla F: \mathbb{X}\to\mathbb{Y}$ is a homeomorphism onto an open convex set $\mathbb{Y}\subset\mathbb{R}^d$. 

Given a list of $N$ vectors $\bm{x}_{1:N}=(\bm{x}_1, \bm{x}_2, \dots, \bm{x}_N)$ and a corresponding list of positive weights $w_{1:N}=(w_1,w_2,\dots,w_N)$ such that $\sum_{i=1}^N w_i=1$, the (left) Bregman barycenter, $\bar{\bm{x}}$, is defined as
\begin{equation}\label{Eq:BregmanBC}
\bar{\bm{x}}(\bm{x}_{1:N},w_{1:N})=\arg\min_{\bm{x}}\sum_{i=1}^N w_iD_F(\bm{x},\bm{x}_i)\,.
\end{equation}
This optimisation problem is easy to solve analytically with the result
\begin{equation}\label{Eq:BregmanAnalytic}
\bar{\bm{x}}(\bm{x}_{1:N},w_{1:N})=\nabla F^{-1}\Big(\sum_{i=1}^N w_i \nabla F(\bm{x}_i) \Big)\,,
\end{equation}
where $\nabla F^{-1}$ is the inverse of the gradient map $\nabla F$. 

\section{Center-seeking agents}\label{Sec:center-seeking_agents}

\begin{definition}[centripetal agent]
Consider a group of $N$ agents with individual views given by the list $\bm{x}_{1:N}$ of vectors from $\mathbb{X}$. We shall call an agent centripetal for the group if she is always willing to change her view to the barycenter $\bar{\bm{x}}(\bm{x}_{1:N},w_{1:N})$. The weights $w_{1:N}$ allow the centripetal agent to express the varying confidence she may have in the abilities of each individual member of the group.
\end{definition}

\begin{theorem}
Consider a group of $N$ centripetal agents with views $\bm{x}_{1:N}$. The agents' weights can be arranged into a square row-stochastic matrix  
\begin{equation}\label{Eq:W}
W=(w_{ij}),\ w_{ij}>0,\ \sum_{j=1}^N w_{ij}=1\ for\ all\ i.
\end{equation}
The group engages in an iterative procedure in which all views are published, then each agent updates their individual view in light of the published information, and the process repeats. The group iteratively converges to a unique consensus view $\bar{\bm{x}}$ which is given by
\begin{equation} \label{Eq:ConsensusX}
\bar{\bm{x}}=\nabla F^{-1}\Big(\sum_{i=1}^N \pi_i \nabla F(\bm{x}_i) \Big)\,,
\end{equation}
where $\pi = (\pi_1,\dots,\pi_N)$ is the stationary distribution of $W$, that is
\begin{equation} \label{Eq:ConsensusPi}
\pi W = \pi, \qquad \sum_{j=1}^N \pi_j = 1, \qquad \pi_j > 0.
\end{equation}

\end{theorem}

\begin{proof}
Let $\bm{x}_{1:N}^{(k)}$ be the list of views for all $N$ agents after the $k$-th iteration starting from $\bm{x}_{1:N}^{(0)}=\bm{x}_{1:N}$. The $i$-th agent computes
\begin{equation}\label{Eq:IterativeProcess}
\bm{x}_i^{(k+1)}= \arg\min_{\bm{x}}\sum_{j=1}^N w_{ij}D_F(\bm{x},\bm{x}_j^{(k)})\,.
\end{equation}
The corresponding analytic solution
\begin{equation}\label{Eq:AnalyticOptimization}
\bm{x}_i^{(k+1)} = \nabla F^{-1}\Big(\sum_{j=1}^N w_{ij} \nabla F(\bm{x}_j^{(k)}) \Big)\,.
\end{equation} 
By applying $\nabla F$ to both sides of this equation and introducing 
\begin{equation}\label{Eq:y}
\bm{y}_i^{(k)}=\nabla F(\bm{x}_i^{(k)})
\end{equation} we obtain 
\begin{equation}
\bm{y}_i^{(k+1)} = \sum_{j=1}^N w_{ij} \bm{y}_j^{(k)}\,,
\end{equation}
and therefore
\begin{equation}\label{Eq:yk}
\bm{y}_i^{(k)} = \sum_{j=1}^N (W^k)_{ij} \bm{y}_j^{(0)}\,.
\end{equation}
Since $W$ is row-stochastic and positive, it is primitive. By the Perron-Frobenius theorem~\cite{HornJohnson_1985}, $W$ has a simple eigenvalue 1 and all other eigenvalues have modulus strictly smaller than 1. Moreover, the powers of $W$ converge
\begin{equation}\label{Eq:limWk}
\lim_{k\to\infty}(W^k)_{ij}=\pi_j\,,
\end{equation} 
where $\pi=(\pi_1, \pi_2, \dots, \pi_N)$ is the unique stationary row vector satisfying
\begin{equation}
\pi W = \pi, \qquad \sum_{j=1}^N \pi_j = 1, \qquad \pi_j > 0.
\end{equation}
Since $\nabla F^{-1}$ is continuous, Eqs.~(\ref{Eq:y}), (\ref{Eq:yk}) and (\ref{Eq:limWk}) imply
\begin{equation}\label{Eq:FinalLimit}
\lim_{k\to\infty}\bm{x}_i^{(k)}=\nabla F^{-1}\Big(\sum_{j=1}^N\pi_j\nabla F(\bm{x}^{(0)}_j)\Big)\,.
\end{equation} 
Note that the right-hand side of this equation does not depend on $i$ and so all agents converge to the same consensus view~(\ref{Eq:ConsensusX}).
\end{proof}
The above proof only requires $W$ to be primitive. The positivity requirement, $w_{ij}>0$, in the statement of the theorem is unnecessarily restrictive. The purpose of that is entirely pedagogical, to make the theorem more intuitive in the first reading.

\section{Emergence of the group agent}\label{Sec:GroupAgent}
We see that the emergent net behaviour of the group converges to that of a single centripetal agent who compiles her view $\bar{\bm{x}}$ as the barycenter
\begin{equation}
\bar{\bm{x}}=\arg\min_{\bm{x}}\sum_{i=1}^N \pi_iD_F(\bm{x},\bm{x}_i)\,.
\end{equation}
While this collective agent does not belong to the original group, she behaves as if she had an opinion regarding the relative abilities of each individual agent. This collective opinion is represented by the stationary vector of weights $\pi$. This happens even when the individual agents are operating completely autonomously and keep their confidence weights private.

\section{Performance}\label{Sec:ConvergenceSpeed}

In the case where the number of agents is small, a hypothetical central planner who knows the individual views and confidence weights would be able to compute the emergent group consensus using Eqs.~(\ref{Eq:ConsensusX})~and~(\ref{Eq:ConsensusPi}). The complexity of doing so is dominated by finding the left eigenvector of $W$ which takes $O(N^3)$ operations in general.

When the agents do not share their confidence weights, or when the number of agents is too large to be processed as a group, we can compute the group's consensus by implementing the iterative procedure~(\ref{Eq:IterativeProcess}). In the rest of this section we estimate the complexity of this approach. 
 
Before we calculate the complexity of the iterative algorithm, we should note that 
the Perron-Frobenius convergence mechanism~(\ref{Eq:limWk}) is used in a great variety of applications. Besides the already mentioned linear DeGroot's model~\cite{DeGroot_1974},
the same mechanism underlies the PageRank algorithm for web ranking~\cite{BrinPage_1998},
the Friedkin-Johnsen model of opinion dynamics with stubborn agents~\cite{FriedkinJohnsen_1990}, 
average-consensus algorithms in distributed multi-agent control~\cite{Olfati-Saber_EtAl_2007}, 
and Markov chain Monte Carlo methods~\cite{GilksEtAl_1996}, among many others. 
In each case, the asymptotic rate of convergence is derived from that of~(\ref{Eq:limWk}), which depends on the second-largest eigenvalue of the iteration matrix $W$~\cite{HornJohnson_1985}.

Because our framework includes the classical linear consensus model as a special case\footnote{Consider $F(\bm{x})=\lVert\bm{x}\rVert^2_2$.}, we start the convergence analysis in the same vein until we encounter and handle the nonlinearity. By Theorem 8.5.1 from~\cite{HornJohnson_1985}, for any $W$ as defined by Eq.~(\ref{Eq:W}) and for any sufficiently small $\epsilon>0$ there exists a constant $C=C(W,\epsilon)$ such that
\begin{equation} \label{HornJohnson_8.5.1}
\lVert W^k-P\rVert_\infty \leq C(|\lambda_2|+\epsilon)^k\ \ {\rm for\ all\ }k=1,2,\dots,
\end{equation}
where $P_{ij}=\pi_j$ and $|\lambda_2|<1$ is the modulus of the second-largest eigenvalue of $W$. 

Let us assemble the $N$ iterates $\bm{y}_i^{(k)}$ defined by Eq.~(\ref{Eq:y}) into a matrix $Y_k\in\mathbb{R}^{N\times d}$. In this notation, Eq.~(\ref{Eq:yk}) becomes $Y_k=W^kY_0$, and we compute
\begin{align}
\lVert Y_k-Y_\infty\rVert_\infty&=\lVert (W^k-P)Y_0\rVert_\infty
                                \leq\lVert W^k-P\rVert_\infty \lVert Y_0\rVert_\infty\cr
                                &\leq C(|\lambda_2|+\epsilon)^k \lVert Y_0\rVert_\infty\ \ {\rm for\ all\ }k=1,2,\dots. \label{Eq:YkInfinity}
\end{align}
For simplicity, let us focus on dynamics away from the boundary $\partial\mathbb{X}$, i.e., we assume that the agents' initial estimates $\bm{x}_{1:N}$, the subsequent iterations $\bm{x}^{(k)}_{1:N}$ and the consensus $\bar{\bm{x}}$ lie in a compact convex $\mathbb{K}\subset\mathbb{X}$. Furthermore, let us assume that the Hessian $\nabla^2 F$ is continuous on $\mathbb{K}$.

Let $m_\mathbb{K}$ be the smallest eigenvalue of $\nabla^2 F$ on $\mathbb{K}$
\begin{equation}
m_\mathbb{K}=\min_{\bm{x}\in\mathbb{K}}\lambda_{\min}\Big(\nabla^2 F(\bm{x})\Big)\,.
\end{equation}
Because $F$ is strictly convex, $\nabla^2 F(\bm{x})$ is positive definite for every $\bm{x}\in\mathbb{K}$. Furthermore, since $\mathbb{K}$ is compact and $\nabla^2 F$ is continuous on $\mathbb{K}$, the minimum eigenvalue is attained and is strictly positive
$m_\mathbb{K}>0$.

Consider a pair of vectors $\bm{x},\bm{x}'\in \mathbb{K}$, $\bm{x}\neq\bm{x}'$. Since $\mathbb{K}$ is convex, any such pair defines a linear segment $[\bm{x},\bm{x}']=\{\bm{x}+t(\bm{x}'-\bm{x})|t\in(0,1)\}\subset\mathbb{K}$. By the mean value theorem one can find a point $\bm{c}\in [\bm{x},\bm{x}']$ such that
\begin{align}
\langle \bm{x}-\bm{x}',\nabla F(\bm{x})-\nabla F(\bm{x}')\rangle & = \langle \bm{x}-\bm{x}',\nabla^2 F(\bm{c})(\bm{x}-\bm{x}')\rangle\cr
                                                & \geq m_\mathbb{K}\lVert \bm{x}-\bm{x}' \rVert_2^2\,.
\end{align}
At the same time, by Cauchy-Schwarz inequality,
\begin{equation}
\langle \bm{x}-\bm{x}',\nabla F(\bm{x})-\nabla F(\bm{x}')\rangle \leq \lVert \bm{x}-\bm{x}' \rVert_2\cdot \lVert \nabla F(\bm{x})-\nabla F(\bm{x}') \rVert_2\,,
\end{equation}
and therefore
\begin{equation}
m_\mathbb{K}\lVert \bm{x}-\bm{x}' \rVert_2\leq \lVert \nabla F(\bm{x})-\nabla F(\bm{x}') \rVert_2\,.
\end{equation}
Applying this to the iterates~(\ref{Eq:y})
\begin{equation}
\lVert \bm{x}_i^{k}-\bar{\bm{x}} \rVert_2\leq \lVert \bm{y}_i^{k}-\bar{\bm{y}} \rVert_2/ m_\mathbb{K}\,,
\end{equation}
where $\bar{\bm{y}}=\nabla F(\bar{\bm{x}})$. We can now combine this with Eq.~(\ref{Eq:YkInfinity}). Because the $L_2$ norm of a single row is always smaller or equal to the infinity norm of an entire matrix, and noticing that $\bar{\bm{y}}=\sum_{j=1}^N\pi_j\bm{y}_j^{(0)}$,
\begin{equation}
\lVert \bm{x}_i^{k}-\bar{\bm{x}} \rVert_2\leq C'(|\lambda_2|+\epsilon)^k\ \ {\rm for\ all\ }k=1,2,\dots,
\end{equation}
where $C'$ is a constant. We see a classic convergence pattern, where each agent rapidly converges to the consensus with the convergence rate determined by the magnitude of the second-largest eigenvalue of $W$. 

Since evaluating $\nabla F$ does not depend on $N$, a single iteration~(\ref{Eq:AnalyticOptimization}) would cost each agent $O(N)$, and the overall effort for each agent is $O(N\ln\epsilon_2/\ln|\lambda_2|)$ where $\epsilon_2$ is the maximum acceptable error in the $L_2$ norm.

\section{Non-scalar means}\label{Sec:NonScalarMeans}
Let $\mathbb{I}\subseteq \mathbb{R}$ be an interval. A scalar mean is defined as a function $M: \mathbb{I}^N\to \mathbb{I}$ such that for all $x_{1:N}\in \mathbb{I}^N$
\begin{equation}
\min(x_1,\dots, x_N)\leq M(x_{1:N}) \leq \max(x_1,\dots, x_N)\,.
\end{equation}
$N$ such means, $M_1,\dots,M_N$, define an iterative mean-type mapping $M_{i,k}$ on the cube $\mathbb{I}^N$ as follows. For the first iteration
\begin{equation}\label{Eq:MTM1}
M_{i,1}=M_i,\ \ \ i=1,2,\dots,N\,,
\end{equation}
and for all subsequent iterations
\begin{equation}\label{Eq:MTM2}
M_{i,k+1}=M_i(M_{1,k}(x_{1:N}),\dots, M_{N,k}(x_{1:N})).
\end{equation}
Matkowski~\cite{Matkowski_1999} proved that under some natural conditions the iterations converge 
\begin{equation}\label{Eq:MatkowskiLimit}
\lim_{k\to\infty}M_{i,k}(x_{1:N})=A(x_{1:N})\,,
\end{equation} 
where the mean $A$ is often called invariant under the mapping~(\ref{Eq:MTM1},\ref{Eq:MTM2}) as it stops responding to subsequent iterations.

The structure of Eq.~(\ref{Eq:BregmanAnalytic}) bears a striking resemblance to that of the scalar quasi-arithmetic means~\cite{Kolmogorov_1930,Nagumo_1930,deFinetti_1931} except, of course, $\bm{x}_{1:N}$ are not scalars. This observation motivated Nielsen~\cite{Nielsen_2023} to extend the definition of the quasi-arithmetic mean to non-scalar types, such as
\begin{equation}\label{Eq:NonScalarMean}
M_{\nabla F}(\bm{x}_{1:N}, w_{1:N})=\nabla F^{-1}\Big(\sum_{i=1}^N w_i \nabla F(\bm{x}_i) \Big)\,.
\end{equation}

We recognize our result as a constructive generalization of Matkowski's convergence theorems to the non-scalar mean~(\ref{Eq:NonScalarMean}) with Eq.~(\ref{Eq:FinalLimit}) being a direct analogue of Eq.~(\ref{Eq:MatkowskiLimit}).

\end{document}